\documentclass[11pt]{article}

\usepackage[margin=0.9in]{geometry}

\usepackage{graphicx}

\usepackage{amssymb}
\usepackage{amsmath}
\usepackage{amsthm}
\usepackage{mathrsfs}
\usepackage[table]{xcolor}
\usepackage{xspace}

\allowdisplaybreaks

\theoremstyle{definition}
\newtheorem{thm}{Theorem}[section]

\newtheorem{proposition}[thm]{Proposition}
\newtheorem{corollary}[thm]{Corollary}

\numberwithin{thm}{section}

\newcommand{\mR}{\mathbb{R}}

\newcommand{\mV}{\mathcal{V}}

\newcommand{\mSS}{\text{SS}}

\newcommand{\mE}{\textbf{E}\xspace}
\newcommand{\mET}{\textbf{ET}\xspace}
\newcommand{\mNP}{\textbf{NP}\xspace}
\newcommand{\mT}{\textbf{T}\xspace}
\newcommand{\mSD}{\textbf{SD}\xspace}

\newcommand{\fulltoday}{\ifcase\month\or
    January\or February\or March\or April\or May\or June\or
    July\or August\or September\or October\or November\or December\fi\space\number\day,\space\number\year}

\begin{document}
\title{Self-Duality and Transfer in Voting Games}
\author{Takaaki Abe\thanks{Department of Economic Engineering, School of Economics, Kyushu University, 744, Motooka, Nishi-ku, Fukuoka, 819-0395, Japan.
Email: takaakiabe@econ.kyushu-u.ac.jp\newline
}}
\date{}
\maketitle

\begin{abstract}
This study examines the role of self-duality in the axiomatization of the Shapley-Shubik power index.
We show that, when self-duality is imposed, the transfer axiom can be weakened to a restricted version requiring transfer only among voting games with no veto player.
This result shows that self-duality can partially substitute for the transfer axiom.
\end{abstract}


\section{Introduction}\label{SEC_INTRO}
The Shapley-Shubik power index is one of the most prominent measures of voting power in simple voting games. Derived from the Shapley value (Shapley, 1953), the index was proposed by Shapley and Shubik (1954) as an application to voting games. In a voting game, a coalition is either winning or losing, depending on whether it can pass a proposal without the cooperation of the remaining players. The Shapley-Shubik power index evaluates each player's power by considering the probability that the player is \textit{pivotal} in a randomly chosen ordering of players. A player $i$ is pivotal in an ordering if the coalition of the players preceding $i$ is losing, whereas it wins once $i$ joins. Thus, the Shapley-Shubik index measures a player's expected marginal ``impact'' on whether a proposal is approved.

Dubey (1975) provided an axiomatic characterization based on the \textit{transfer} axiom. This axiom is based on the observation that the set of monotonic voting games forms a lattice. The axiom requires that, for any two voting games, the sum of the evaluations of their join and meet coincides with the sum of their original evaluations. Dubey et al. (2005) showed that this axiom can be paraphrased as follows: the change in power caused by adding a collection of winning coalitions is independent from the initial game to which these coalitions are added.
Einy and Haimanko (2011) showed that the \textit{efficiency} axiom can be replaced with the \textit{gain-loss} property if the null player property is strengthened to the dummy player property.\footnote{The gain-loss property requires that whenever one player's power increases, another player's power must decrease. Their axiomatization consists of Gain-Loss, Equal Treatment, Dummy Player, and Transfer. See Einy and Haimanko (2011) for details.} These results illustrate that the transfer axiom plays a central role in characterizing the Shapley-Shubik index, while also raising the question of how much transfer is actually needed.

We address this question by introducing \textit{self-duality} (Sudh\"{o}lter, 1997; Funaki, 1998).\footnote{A related form of self-duality appears in Young (1994) in the context of bankruptcy problems; see Chapter 4, Section 7.} For a game $v$, the dual game $d_v$ is given as follows: for every $S\subseteq N$, $d_v(S)=v(N)-v(N\setminus S)$. In the context of voting games, the dual game has a natural interpretation: a coalition is winning in $d_v$ if and only if it can block its complement $N\setminus S$ from winning in the original game $v$. Thus, $v$ describes the capacity of each coalition to approve a proposal, whereas $d_v$ describes its capacity to block the complementary coalition. A power index satisfies self-duality if it assigns the same power in both representations.

As its main contribution, this study shows that self-duality permits a weakening of the transfer axiom. We require transfer only within the class of voting games with no veto player. We call this restricted requirement \textit{veto-free transfer}. We show that, together with efficiency, the equal treatment property, and the null player property, self-duality and veto-free transfer characterize the Shapley-Shubik power index.

The remainder of this paper is organized as follows.
Section \ref{SEC_PREL} introduces the basic definitions and notation.
Section \ref{SEC_AX} presents the axioms and the characterization result.
Section \ref{SEC_CONC} presents the concluding remarks.
Section \ref{SEC_PROOF} provides the proofs.

\section{Preliminaries}\label{SEC_PREL}
Let $N = \{1,..., n\}$ be the set of players, where $n\geq 3$.\footnote{When $n=2$, the Shapley-Shubik index is uniquely determined by Efficiency, Equal Treatment, and Null Player introduced below.}
Let $S \subseteq N$ denote a coalition.
Let $2^N$ represent the set of all coalitions.
A \textit{game} on $N$ is defined as a function $v : 2^N \rightarrow \mR$ with $v(\emptyset) = 0$.
A game $v$ is called a \textit{voting game} if it satisfies the following properties:
\begin{itemize}
    \item $v(S) \in \{0,1\}$ for every $S \subseteq N$;
    \item $v(N) = 1$;
    \item $v(S) \leq v(T)$ whenever $S \subseteq T \subseteq N$.
\end{itemize}

Let $\mV_N$ denote the set of all voting games on $N$. A coalition $S \subseteq N$ is said to be \textit{winning} in $v \in \mV_N$ if $v(S) = 1$, and \textit{losing} otherwise.
A winning coalition $S\subseteq N$ in $v$ is called a \textit{minimal winning coalition} if no proper subset of $S$ is winning in $v$; that is, $v(S)=1$ and $v(T)=0$ for every $T\subsetneq S$.
Let $W_v$ denote the set of all winning coalitions of $v$.
Let $M_v$ denote the set of all minimal winning coalitions of $v$.
For any finite set $X$, let $|X|$ denote its cardinality.
Players $i$ and $j$ are said to be \textit{symmetric} in $v$ if for every $S\subseteq N\setminus \{i,j\}$, $v(S\cup \{i\})=v(S\cup \{j\})$. Player $i\in N$ is called a \textit{null player} in $v$ if for every $S\subseteq N\setminus \{i\}$, $v(S\cup \{i\})-v(S)=0$.
Player $i$ is called a \textit{veto player} in $v$ if $v(S)=0$ for every $S\subseteq N\setminus \{i\}$.
For every nonempty coalition $T\subseteq N$, the \textit{unanimity game}
$u_T\in \mV_N$ is defined as follows: for every $S\subseteq N$,
\[
u_T(S)=
\begin{cases}
1 & \text{if } T\subseteq S,\\
0 & \text{otherwise}.
\end{cases}
\]
For every $i\in N$, $u_{\{i\}}$ is called a \textit{dictator game}.

A \textit{power index} is a function $\phi: \mV_N \rightarrow \mR^n$. For every $v \in \mV_N$ and $i \in N$, the value $\phi_i(v)$ represents the voting power of player $i$ in the voting game $v$.
The Shapley-Shubik power index $\mSS : \mV_N \rightarrow \mR^n$ is given as follows: for every $v \in \mV_N$ and $i \in N$,
\[
\mSS_i(v) = \sum_{S \subseteq N \setminus \{i\}} \frac{|S|! \, (|N| - |S| - 1)!}{|N|!} \left( v(S \cup \{i\}) - v(S) \right).
\]

\section{Axioms and Results}\label{SEC_AX}
\subsection{Axioms}
We begin by introducing standard axioms for power indices.
\begin{itemize}
\item[] \textbf{Efficiency (\mE).}
For every $v \in \mV_N$, $\sum_{j\in N}\phi_j(v)=1$.

\item[] \textbf{Equal Treatment (\mET).}
For every $v \in \mV_N$ and $i,j \in N$, if $i$ and $j$ are symmetric in $v$, then $\phi_i(v)=\phi_j(v)$.

\item[] \textbf{Null Player (\mNP).}
For every $v \in \mV_N$ and $i \in N$, if $i$ is a null player in $v$, then $\phi_i(v)=0$.
\end{itemize}
The efficiency axiom can be viewed as a normalization requirement: it states that the players' power shares sum to one. The equal treatment property requires that symmetric players receive the same evaluation. The null player property states that a null player receives zero.

We now introduce the axiom studied by Dubey (1975). For every $v,w\in \mV_N$ and $S\subseteq N$, define
\begin{itemize}
\item[] $(v \vee w)(S)= \max\{v(S),w(S)\}$,
\item[] $(v \wedge w)(S)= \min\{v(S),w(S)\}$.
\end{itemize}
In words, a coalition is winning in $v \vee w$ if it is winning in $v$ or $w$, while it is winning in $v \wedge w$ if it is winning in $v$ and $w$.
Note that $\mV_N$ is a lattice: $\mV_N$ is closed under operations $\vee$ and $\wedge$.
Moreover, $\mV_N$ is a distributive lattice: for every $v_1,\ldots,v_m,w\in \mV_N$, $(v_1\vee\cdots\vee v_m)\wedge w = (v_1\wedge w)\vee\cdots\vee(v_m\wedge w)$.\footnote{This is immediately checked as follows.
For every $S\subseteq N$, $((v_1\vee\cdots\vee v_m)\wedge w)(S) = \min\{\max_{\ell=1,\ldots,m}v_\ell(S),w(S)\}$.
If $w(S)=0$, then we have $\min\{\max_{\ell=1,\ldots,m}v_\ell(S),w(S)\}=0=\max_{\ell=1,\ldots,m}\min\{v_\ell(S),w(S)\}$.
If $w(S)=1$, then we have $\min\{\max_{\ell=1,\ldots,m}v_\ell(S),w(S)\}=\max_{\ell=1,\ldots,m}v_\ell(S)=\max_{\ell=1,\ldots,m}\min\{v_\ell(S),w(S)\}$.
Therefore, for every $S\subseteq N$,
$((v_1\vee\cdots\vee v_m)\wedge w)(S)=
\min\{\max_{\ell=1,\ldots,m}v_\ell(S),w(S)\}=
\max_{\ell=1,\ldots,m}\min\{v_\ell(S),w(S)\}=
((v_1\wedge w)\vee\cdots\vee(v_m\wedge w))(S)$.
}
Dubey (1975) introduced the following axiom.
\begin{itemize}
\item[] \textbf{Transfer (\mT).}
For every $v,w \in \mV_N$, $\phi(v \vee w)+\phi(v \wedge w)=\phi(v)+\phi(w)$.
\end{itemize}
Dubey et al. (2005) state that the transfer axiom is equivalent to the following statement:
for every $v,v',w,w'\in \mV_N$, if $W_{v'}\subseteq W_{v}$, $W_{w'}\subseteq W_{w}$, and
$W_v\setminus W_{v'} = W_w\setminus W_{w'}$, then
$\phi(v)-\phi(v') = \phi(w)-\phi(w')$.
In other words, whenever the same collection of winning coalitions is added, the resulting change in power is the same.

Now, we introduce the notion of self-duality.
For every $v \in \mV_N$, the \textit{dual game} $d_v$ of $v$ is defined by
\[
d_v(S) = v(N) - v(N \setminus S)
\]
for every $S \subseteq N$.
For every $A\subseteq \mV_N$, define $d(A)=\{d_v\mid v\in A\}$.
The notion of a dual game was originally defined for general TU games (see Oishi et al. (2016) for a comprehensive review of duality). The worth $d_v(S)$ can be interpreted as the amount that remains for coalition $S$ after assigning to its complement $N \setminus S$ the amount it (namely, $N \setminus S$) claims in the original game $v$.

In the context of voting games, this definition has a particularly appealing interpretation.
For every $v\in \mV_N$, a winning coalition in $d_v$ can be seen as a ``blocking coalition'' in $v$, since $d_v(S)=1$ if and only if $v(N\setminus S)=0$.
More specifically, a coalition $S$ is winning in the dual game $d_v$ if and only if the coalition $S$ can block $N\setminus S$ from being winning in $v$.
Note that $\mV_N$ is closed under the self-dual operation: for every $v\in \mV_N$, $d_v\in  \mV_N$.\footnote{This is immediately checked as follows. We have $d_v(N)=v(N)-v(\emptyset)=1$ and $d_v(\emptyset)=v(N)-v(N)=0$. Now, let $S,T\subseteq N$ with $S\subseteq T$. Then, $N\setminus T \subseteq N\setminus S$. By monotonicity of $v$, $v(N\setminus T)\leq v(N\setminus S)$. Hence, $
d_v(S)=1-v(N\setminus S)\leq 1-v(N\setminus T)=d_v(T)$. Therefore, $d_v\in \mV_N$.}

\begin{itemize}
\item[] \textbf{Self-Duality (\mSD).}
For every $v \in \mV_N$, $\phi(v)=\phi(d_v)$.
\end{itemize}
Self-duality states that a player's power is invariant whether it is measured via winning or blocking coalitions.\footnote{Sudh\"{o}lter (1997) shows that the \textit{modified nucleolus} satisfies \mSD for general TU games; see (iii) of Remark 1.2.}

\subsection{Axiomatization}

We show that self-duality can \textit{partially} substitute for the transfer axiom.
Let $\mV^*_N$ denote the set of voting games with no veto player.
We begin by examining the structure of $\mV^*_N$ and the class of dual games obtained from $\mV^*_N$.

\begin{proposition}\label{PROP_STAR_LATTICE}
$\mV^*_N$ is a lattice.
\end{proposition}
Proposition \ref{PROP_STAR_LATTICE} states that if $v$ and $w$ have no veto players, then both $v \vee w$ and $v \wedge w$ also have no veto players.
Moreover, the following proposition ensures that taking dual games carries lattice structure over to the dual class.

\begin{proposition}\label{PROP_DUAL_IS_LATTICE}
Let $A\subseteq \mV_N$ be a nonempty lattice. Then, $d(A)$ is also a lattice.
\end{proposition}

Combining Propositions \ref{PROP_STAR_LATTICE} and \ref{PROP_DUAL_IS_LATTICE}, we immediately obtain the following corollary.

\begin{corollary}\label{COR_PLUS_LATTICE}
Let $\mV^+_N := d(\mV^*_N)$. $\mV^+_N$ is a lattice.
\end{corollary}

Each element of $\mV^+_N$ is a game with no one-person minimal winning coalition.
For example, no dictator game belongs to $\mV^+_N$. Similarly, a game $v$ with $M_v=\{\{1\},\{2\}\}$ also does not belong to $\mV^+_N$.
The following result shows that $\mV^*_N$, $\mV^+_N$, and the dictator games together cover the entire domain $\mV_N$.

\begin{proposition}\label{PROP_DECOMP}
$\mV_N= \mV^*_N \cup \mV^+_N \cup \{u_{\{i\}} | i\in N\}$.
\end{proposition}

Given the duality relationship between $\mV^*_N$ and $\mV^+_N$, Proposition \ref{PROP_DECOMP} suggests that the full transfer axiom can be weakened in the presence of self-duality. This observation motivates the following weaker axiom, which requires transfer only for veto-free voting games.
\begin{itemize}
\item[] \textbf{Veto-Free Transfer ($\mT^*$).}
For every $v,w \in \mV^*_N$, $\phi(v \vee w)+\phi(v \wedge w)=\phi(v)+\phi(w)$.
\end{itemize}

\begin{proposition}\label{PROP_CHARACTERIZATION}
A function $\phi: \mV_N \rightarrow \mR^n$ satisfies \mE, \mET, \mNP, \mSD, $\mT^*$ if and only if $\phi=\mSS$.
\end{proposition}

Proposition \ref{PROP_CHARACTERIZATION} shows that self-duality can partially substitute for the role of the transfer axiom. 
Indeed, although transfer is imposed only on the veto-free domain $\mV^*_N$, self-duality extends its role to the dual class $\mV^+_N$. 
Nevertheless, the two axioms are logically independent. The following functions show that the axioms used in Proposition \ref{PROP_CHARACTERIZATION} are independent. Note that $n\geq 3$.

\begin{itemize}
\item Function $f^{\neg\mE}(v):= \frac{1}{2}\mSS(v)$ violates \mE but satisfies the other four axioms.
\item Let $\lambda(k):=\{1,...,k-1\}$ for $k=2,...,n$, where $\lambda(1)=\emptyset$.
Let $\rho(k):=\{k+1,...,n\}$ for $k=1,...,n-1$, where $\rho(n)=\emptyset$.
For every $i\in N$, define $f^{\neg\mET}_i(v)=\frac{1}{2}\left( v(\lambda(i)\cup\{i\})-v(\lambda(i)) \right)+\frac{1}{2}\left( v(\rho(i)\cup\{i\})-v(\rho(i)) \right)$. This function violates \mET but satisfies the other four axioms.

\item Function $f^{\neg\mNP}_i(v):= \frac{1}{n}$ violates \mNP but satisfies the other four axioms.

\item Let $\tilde{v}$ be the game with $M_{\tilde{v}}=\{\{1,k\} \mid k\in\{2,...,n\}\}$.
Define $f^{\neg\mSD}_i(\tilde{v}):= \frac{1}{n}$ and $f^{\neg\mSD}_i(v):= \mSS_i(v)$ for every $v\in\mV_N\setminus \{\tilde{v}\}$. Note that there is no null player in $\tilde{v}$ and that $\tilde{v}\notin \mV^*_N$ as $1$ is a veto player. Moreover, $d_{\tilde{v}}\in \mV^*_N$ as $M_{d_{\tilde{v}}}=\{\{1\},\{2,...,n\}\}$. This function violates \mSD but satisfies the other four axioms.

\item Define $f^{\neg\mT^*}_i(\tilde{v}):= \frac{1}{n}$, $f^{\neg\mT^*}_i(d_{\tilde{v}}):= \frac{1}{n}$, and $f^{\neg\mT^*}_i(v):= \mSS_i(v)$ for every $v\in\mV_N\setminus \{\tilde{v}, d_{\tilde{v}}\}$. Note that there is no null player in $d_{\tilde{v}}$. This function violates $\mT^*$ but satisfies the other four axioms.

\end{itemize}

\section{Conclusion}\label{SEC_CONC}
This study provided a characterization of the Shapley-Shubik power index by combining self-duality with a restricted form of the transfer axiom. 
While the transfer axiom is usually imposed on the entire domain of voting games, we showed that, in the presence of self-duality, it is sufficient to require transfer only on the class of voting games without veto players.
This result shows that self-duality can partially substitute for transfer, although self-duality and veto-free transfer are logically independent.

\section{Proofs}\label{SEC_PROOF}

\noindent\textbf{Proposition \ref{PROP_STAR_LATTICE}.}
$\mV^*_N$ is a lattice.
\begin{proof}
We show that $\mV^*_N$ is closed under the operations $\vee$ and $\wedge$.
Let $v,w\in \mV^*_N$. 

First, we show that $v\vee w$ has no veto player.
Let $i\in N$. Since $v\in \mV_N^*$, player $i$ is not a veto player in $v$. Hence, there is a coalition $S\subseteq N\setminus\{i\}$ such that $v(S)=1$.
Then, since $w(S)\in \{0,1\}$,
\[
(v\vee w)(S)=\max\{v(S),w(S)\}=1.
\]
Since $S\subseteq N\setminus\{i\}$, $i$ is not a veto player in $v\vee w$.
This holds for every $i\in N$. Hence, $v\vee w$ has no veto player: formally, $v\vee w\in \mV_N^*$.

Next, we show that $v\wedge w$ has no veto player.
Let $i\in N$. Since $v\in \mV_N^*$, player $i$ is not a veto player in $v$. Hence, there is a coalition $S\subseteq N\setminus\{i\}$ such that $v(S)=1$.
Similarly, since $w\in \mV_N^*$, player $i$ is not a veto player in $w$. Hence, there is a coalition $T\subseteq N\setminus\{i\}$ such that $w(T)=1$.
Since $S,T \subseteq N\setminus\{i\}$, we have $S\cup T\subseteq N\setminus\{i\}$.
By monotonicity of $v$ and $w$, we obtain $v(S\cup T)\geq v(S)=1$ and $w(S\cup T)\geq w(T)=1$. Hence, $v(S\cup T)=1$ and $w(S\cup T)=1$.
It follows that
\[
(v\wedge w)(S\cup T)=\min\{v(S\cup T),w(S\cup T)\}=1.
\]
Since $S\cup T \subseteq N\setminus\{i\}$, $i$ is not a veto player in $v\wedge w$.
This holds for every $i\in N$. Hence, $v\wedge w$ has no veto player: formally, $v\wedge w\in \mV_N^*$.
\end{proof}

\noindent\textbf{Lemma A.}
For every $v,w\in \mV_N$, $d_{v\vee w}=d_v\wedge d_w$ and $d_{v\wedge w}=d_v\vee d_w$.
\begin{proof}
For every $S\subseteq N$, we have
$d_{v\vee w}(S)
= 1-(v\vee w)(N\setminus S)
= 1-\max\{v(N\setminus S),w(N\setminus S)\}
= \min\{1-v(N\setminus S),1-w(N\setminus S)\}
= \min\{d_v(S),d_w(S)\}
= (d_v\wedge d_w)(S)$.
Similarly, for every $S\subseteq N$, we have
$d_{v\wedge w}(S)
= 1-(v\wedge w)(N\setminus S)
= 1-\min\{v(N\setminus S),w(N\setminus S)\}
= \max\{1-v(N\setminus S),1-w(N\setminus S)\}
= \max\{d_v(S),d_w(S)\}
= (d_v\vee d_w)(S)$.
\end{proof}

\noindent\textbf{Proposition \ref{PROP_DUAL_IS_LATTICE}.}
Let $A\subseteq \mV_N$ be a nonempty lattice. Then, $d(A)$ is also a lattice.
\begin{proof}
Let $A\subseteq \mV_N$ be a nonempty lattice. Let $v',w'\in d(A)$.
We show that $v'\vee w' \in d(A)$ and $v'\wedge w'\in d(A)$.
There are $v,w\in A$ such that $v'=d_v$ and $w'=d_w$.
Since $A$ is a lattice, we have $v\wedge w\in A$.
Hence, $d_{v\wedge w} \in d(A)$.
We have $v'\vee w' = d_v\vee d_w \overset{\text{Lma.A}}{=} d_{v\wedge w} \in d(A)$.
Similarly, $v\vee w\in A$. Hence, $d_{v\vee w} \in d(A)$.
We have $v'\wedge w'=d_v\wedge d_w \overset{\text{Lma.A}}{=}d_{v\vee w} \in d(A)$.
This completes the proof.
\end{proof}

\noindent\textbf{Proposition \ref{PROP_DECOMP}.}
$\mV_N= \mV^*_N \cup \mV^+_N \cup \{u_{\{i\}} | i\in N\}$.
\begin{proof}
It is straightforward that $\mV^*_N \cup \mV^+_N \cup \{u_{\{i\}} | i\in N\} \subseteq \mV_N$.
We show the reverse inclusion.
Let $v\in \mV_N$. If $v$ has no veto player, then $v\in \mV_N^*$.
Hence, suppose that $v$ has at least one veto player.
We consider two cases: (i) there is $i\in N$ such that $v(\{i\})=1$; (ii) for every $i\in N$, $v(\{i\})=0$.

Case (i). Since $v(\{i\})=1$, for every $S\subseteq N$ with $i\in S$, $v(S)=1$. Now assume that there is $S\subseteq N\setminus \{i\}$ such that $v(S)=1$. Then, $\cap_{S\in W_v}S$ is empty, which contradicts the fact that $v$ has a veto player.
Hence, for every $S\subseteq N\setminus \{i\}$, $v(S)=0$.
Therefore, $v=u_{\{i\}}$.

Case (ii). For this $v$, consider $d_v$. Assume that $d_v$ has a
veto player $i\in N$. Then, for every $S\subseteq N\setminus\{i\}$,
we have $d_v(S)=0$.
Since $d_v(S)=1-v(N\setminus S)$ for every $S\subseteq N$, it holds that for every $S\subseteq N\setminus\{i\}$, $v(N\setminus S)=1$.
Taking $S=N\setminus\{i\}$, we obtain
$v(\{i\})= v(N\setminus (N\setminus \{i\})) = 1$,
which contradicts the assumption of Case (ii).
Therefore, $d_v$ has no veto player: equivalently, $d_v\in \mV_N^*$.
It follows that $v=d_{(d_v)}\in \mV_N^+$.
\end{proof}

\noindent\textbf{Lemma B1.}
For every $v\in \mV_N^+$ and $T\in M_{v}$, $u_T \in \mV_N^+$.
\begin{proof}
We first show that for every $v \in \mV_N^+$ and $i\in N$, $v(\{i\})=0$.
Let $v\in \mV_N^+$. By definition of $\mV_N^+$, there is $v'\in \mV_N^*$ such that $v=d_{v'}$.
Let $i\in N$. Since $v'\in \mV_N^*$, the player $i$ is not a veto player in $v'$.
Hence, there is a coalition $S\subseteq N\setminus\{i\}$ such that
$v'(S)=1$. Since $v'$ is monotonic, we have $v'(N\setminus\{i\})=1$, which implies that $v(\{i\}) = d_{v'}(\{i\}) = 1-v'(N\setminus\{i\}) = 0$. This holds for every $i\in N$.
Hence, for every $v \in \mV_N^+$ and $i\in N$, $v(\{i\})=0$.
This implies that every minimal winning coalition of $v$ contains at least two players: formally, for every $T\in M_v$, $|T|\geq 2$.

Now, we show that for every $T\subseteq N$ with $|T|\geq 2$, $u_T \in \mV_N^+$.
Let $T\subseteq N$ with $|T|\geq 2$.
For the given $T$, let $w$ be the voting game satisfying $M_w=\{\{i\}\mid i\in T\}$.
Then, $w(S)=1$ if and only if $S\cap T\neq\emptyset$.
Since $|T|\geq 2$, no player is a veto player in $w$; hence $w\in \mV_N^*$.
For every $S\subseteq N$, we have the following:
\[
\begin{aligned}
d_w(S)=1
&\iff w(N\setminus S)=0\\
&\iff (N\setminus S)\cap T=\emptyset\\
&\iff T\subseteq S\\
&\iff u_T(S)=1.
\end{aligned}
\]
Hence, $d_w=u_T$. Since $w\in \mV_N^*$, we have $d_w\in \mV_N^+$. Therefore, $u_T\in \mV_N^+$.
\end{proof}

\noindent\textbf{Lemma B2.}
If $\phi$ satisfies \mSD and $\mT^*$, then for every $v,w\in \mV_N^+$,
$\phi(v\vee w)+\phi(v\wedge w)=\phi(v)+\phi(w)$.
\begin{proof}
Let $v,w\in \mV_N^+$. 
Note that $v\vee w \in \mV_N^+$ and $v\wedge w \in \mV_N^+$ by Corollary \ref{COR_PLUS_LATTICE}.
By definition of $\mV_N^+$, there are $v',w'\in \mV_N^*$ such that $v=d_{v'} \text{ and } w=d_{w'}$.
Therefore, we have
\[
\begin{alignedat}{2}
\phi(v\vee w)+\phi(v\wedge w)
&= \phi(d_{v'}\vee d_{w'})+\phi(d_{v'}\wedge d_{w'})\\
&= \phi(d_{v'\wedge w'})+\phi(d_{v'\vee w'})
&\quad& \text{by Lemma A}\\
&= \phi(v'\wedge w')+\phi(v'\vee w')
&\quad& \text{by \mSD}\\
&= \phi(v'\vee w')+\phi(v'\wedge w')\\
&= \phi(v')+\phi(w')
&\quad& \text{by $\mT^*$ and $v',w'\in \mV_N^*$}\\
&= \phi(d_{v'})+\phi(d_{w'})
&\quad& \text{by \mSD}\\
&= \phi(v)+\phi(w).
\end{alignedat}
\]
\end{proof}

\noindent\textbf{Lemma B3.}
Let $A\subseteq \mV_N$ be a nonempty lattice. Suppose that $\phi:\mV_N\rightarrow \mR^n$ satisfies the following: for every $v,w\in A$, 
$\phi(v\vee w)+\phi(v\wedge w)=\phi(v)+\phi(w)$.
Then, for every $k\geq 1$ and every $v_1, ...,v_k\in A$,
\[
\phi(v_1\vee ...\vee v_k)
=
\sum_{\emptyset\neq I\subseteq \{1,...,k\}}
(-1)^{|I|+1}
\phi\left(\bigwedge_{\ell\in I}v_\ell\right).
\]
\begin{proof}
Note that this result can be viewed as a slight generalization of Lemma 2.3 of Einy (1987) to an arbitrary sublattice $A$. Since the proof of Lemma 2.3 in Einy (1987) is omitted, we provide a proof below for the sake of completeness.

We prove the equation by induction on $k$.
For $k=1$, the equation immediately holds.
Let $k\geq 2$. Suppose that the equation holds for $k-1$.
Let $v_1, ...,v_k\in A$ and $w=v_1\vee ...\vee v_{k-1}$.
Since $A$ is a lattice, we have $w\in A$.
Moreover, since $\phi$ satisfies $\phi(v'\vee w')+\phi(v'\wedge w')=\phi(v')+\phi(w')$ for every $v',w'\in A$, we obtain
\begin{equation}
\phi(w\vee v_k)=\phi(w)+\phi(v_k)-\phi(w\wedge v_k). \label{eq_0524_1252}
\end{equation}
We focus on the right-hand side of (\ref{eq_0524_1252}).
For the term $\phi(w)$, by the induction hypothesis,
\begin{equation}
\phi(w) = \sum_{\emptyset\neq I\subseteq \{1,...,k-1\}}
(-1)^{|I|+1} \phi\left(\bigwedge_{\ell\in I}v_\ell\right). \label{eq_0524_1253}
\end{equation}
We next consider the term $\phi(w\wedge v_k)$ of (\ref{eq_0524_1252}). By the distributive rule of $\vee$ and $\wedge$ on $\mV_N$, we have
$w\wedge v_k = (v_1\vee...\vee v_{k-1})\wedge v_k = (v_1\wedge v_k)\vee...\vee(v_{k-1}\wedge v_k)$.
Since $A$ is a lattice, each $v_\ell\wedge v_k$ belongs to $A$.
By the induction hypothesis, we obtain
\begin{align}
\phi(w\wedge v_k)
&=\phi((v_1\wedge v_k)\vee...\vee(v_{k-1}\wedge v_k)) \nonumber\\
&=\sum_{\emptyset\neq I\subseteq \{1,...,k-1\}}
(-1)^{|I|+1} \phi\left(\bigwedge_{\ell\in I}(v_\ell\wedge v_k)\right) \nonumber\\
&=\sum_{\emptyset\neq I\subseteq \{1,...,k-1\}}
(-1)^{|I|+1} \phi\left(\bigwedge_{\ell\in I\cup\{k\}}v_\ell\right).\label{eq_0524_1305}
\end{align} 
Since $w\vee v_k=v_1\vee...\vee v_k$, we have the following:
\begin{align}
\phi(v_1\vee...\vee v_k)
&=
\phi(w\vee v_k)\nonumber\\
&\overset{(\ref{eq_0524_1252})}{=} \phi(w)+\phi(v_k)-\phi(w\wedge v_k) \nonumber\\
&\overset{(\ref{eq_0524_1253})(\ref{eq_0524_1305})}{=}
\sum_{\emptyset\neq I\subseteq \{1,...,k-1\}}
(-1)^{|I|+1}
\phi\left(\bigwedge_{\ell\in I}v_\ell\right)
+\phi(v_k) \nonumber\\
&\quad
-\sum_{\emptyset\neq I\subseteq \{1,...,k-1\}}
(-1)^{|I|+1}\phi\left(\bigwedge_{\ell\in I\cup\{k\}}v_\ell \right) \nonumber\\
&=
\sum_{\emptyset\neq I\subseteq \{1,...,k-1\}}
(-1)^{|I|+1}\phi\left(\bigwedge_{\ell\in I}v_\ell\right)
+\phi(v_k) \nonumber\\
&\quad
+\sum_{\emptyset\neq I\subseteq \{1,...,k-1\}}
(-1)^{|I|+2}
\phi\left(\bigwedge_{\ell\in I\cup\{k\}}v_\ell\right)  \nonumber\\
&=
\sum_{\substack{\emptyset\neq I\subseteq \{1,...,k\} \\ k\notin I}}
(-1)^{|I|+1}\phi\left(\bigwedge_{\ell\in I}v_\ell\right)
+\phi(v_k) \nonumber\\
&\quad
+\sum_{\substack{ I\subseteq \{1,...,k\} \\ k\in I \\ I\setminus\{k\}\neq \emptyset}}
(-1)^{|I|+1}
\phi\left(\bigwedge_{\ell\in I}v_\ell\right) \nonumber\\
&=\sum_{\emptyset\neq I\subseteq \{1,...,k\}}
(-1)^{|I|+1}
\phi\left(\bigwedge_{\ell\in I}v_\ell\right). \nonumber
\end{align}
This completes the proof.
\end{proof}

\noindent\textbf{Proposition \ref{PROP_CHARACTERIZATION}.}
A function $\phi: \mV_N \rightarrow \mR^n$ satisfies \mE, \mET, \mNP, \mSD, $\mT^*$ if and only if $\phi=\mSS$.
\begin{proof}
It is straightforward that $\mSS$ satisfies \mE, \mET, \mNP, \mSD, and $\mT^*$. We prove the uniqueness part.
Let $v\in \mV_N$. By Proposition \ref{PROP_DECOMP}, we have (i) $v\in \{u_{\{i\}} | i\in N\}$,  (ii) $v\in \mV^+_N$, or (iii) $v\in \mV^*_N$.

Case (i). By \mE and \mNP, $\phi(v)$ is uniquely determined.

Case (ii). Since every voting game can be written as the join of the unanimity games associated with its minimal winning coalitions, we have $v=u_{T_1}\vee ... \vee u_{T_k}$, where $T_1, ..., T_k\in M_v$.
Moreover, for every $\ell\in \{1,...,k\}$, Lemma B1 implies that $u_{T_\ell}\in \mV^+_N$.
In addition, since $\phi$ satisfies \mSD and $\mT^*$, Lemma B2 ensures that $\phi(v\vee w)+\phi(v\wedge w)=\phi(v)+\phi(w)$ for every $v,w\in \mV_N^+$.
Since $\bigwedge_{\ell\in I}u_{T_\ell} = u_{\cup_{\ell\in I}T_\ell}$ for every nonempty $I\subseteq \{1,...,k\}$, Lemma B3 implies that
\[
\phi(v)= \sum_{\emptyset\neq I\subseteq \{1,...,k\}} (-1)^{|I|+1} \phi(u_{\cup_{\ell\in I}T_\ell}).
\]
For every $T\subseteq N$, $\phi(u_{T})$ is uniquely determined by \mE, \mET, and \mNP.
Hence, $\phi(v)$ is uniquely determined.

Case (iii). It holds that $d_v\in \mV^+_N$. By \mSD, we have $\phi(v)=\phi(d_v)$, where, by Case (ii), $\phi(d_v)$ is uniquely determined.

Since $\mSS$ satisfies the axioms, the uniquely determined function $\phi$ must coincide with $\mSS$.
\end{proof}

\section*{Declaration}

\noindent\textbf{Funding.} The author gratefully acknowledges the financial support from JSPS: No.22H00829.

\noindent\textbf{Conflict of interest.} The author declares that there are no conflicts of interest.

\noindent\textbf{Data availability.} This work is not based on any empirical or simulated data.

\end{document}